\documentclass[11pt]{article}

\usepackage{fullpage}

\usepackage[dvipsnames,usenames]{xcolor}

\usepackage{subfig}
\usepackage{amsthm,amssymb,amsmath}  
\usepackage{xspace,enumerate}
\usepackage[utf8]{inputenc}
\usepackage{thmtools}
\usepackage{thm-restate}
\usepackage{blkarray}
\usepackage{chngcntr}

\usepackage{authblk}
\usepackage[hypertexnames=false,colorlinks=true,urlcolor=Blue,citecolor=Green,linkcolor=BrickRed]{hyperref}

\usepackage[ruled]{algorithm}
\usepackage[noend]{algpseudocode}
\algnewcommand{\LineComment}[1]{\State \(\triangleright\) #1}

\usepackage{graphicx}
\usepackage[noadjust]{cite}
  \usepackage{microtype}
    \usepackage{xspace}
    \usepackage{amsfonts}
    \usepackage{wasysym}
  \usepackage{makecell}
  \usepackage{tabularx}
  \usepackage{comment}
  \usepackage{todonotes}
  \usepackage{thmtools}
  \usepackage{thm-restate}
  \usepackage[capitalise]{cleveref}
  \usepackage{verbatim}
\usepackage{units}
\usepackage{color}
\definecolor{darkblue}{rgb}{0,0.08,0.45}
  \usepackage{epsfig}
    \usepackage{graphics}
\usepackage{pgf,tikz}
  \usepackage{graphicx}

\usetikzlibrary{trees, arrows, shapes, snakes}
\usepackage{marvosym}

\usetikzlibrary{decorations.pathmorphing}
\usetikzlibrary{decorations.markings}
\usetikzlibrary{decorations.pathmorphing,shapes}
\usetikzlibrary{calc,decorations.pathmorphing,shapes}
\usepackage{booktabs}
\usepackage{etoolbox}
\usepackage{enumitem}

\newcommand{\cOtilde}{\tilde{O}}

\newcommand{\ignore}[1]{}

\newcommand{\fp}{\ensuremath\phi}
\theoremstyle{plain}
\newtheorem{theorem}{Theorem}
\newtheorem{lemma}[theorem]{Lemma} 
\newtheorem{corollary}[theorem]{Corollary}  
\newtheorem{proposition}[theorem]{Proposition}

\newtheorem*{conjecture*}{Conjecture}

\author[1]{Shay Mozes}
\author[2]{Nathan Wallheimer}
\author[2]{Oren Weimann}

\affil[1]{
The Interdisciplinary Center Herzliya, Israel\\
\href{mailto:smozes@idc.ac.il}{smozes@idc.ac.il}}

\affil[2]{
University of Haifa, Israel\\
\href{mailto:nathanwallh@gmail.com}{nathanwallh@gmail.com}, \href{mailto:oren@cs.haifa.ac.il}{oren@cs.haifa.ac.il}
}

\date{}

\title{Improved Compression of the Okamura-Seymour Metric}

\begin{document}

\maketitle

\begin{abstract}
\noindent Let $G=(V,E)$ be an undirected unweighted planar graph.
Consider a vector storing the distances from an arbitrary vertex $v$ to all vertices $S = \{ s_1 , s_2 , \ldots , s_k \}$ of a single face in their cyclic order. The {\em pattern} of $v$ is obtained by taking the difference between every pair of consecutive values of this vector. In STOC'19, Li and Parter used a VC-dimension argument to show that in planar graphs, the number of {\em distinct} patterns, denoted $x$, is only $O(k^3)$. 
This resulted in a simple compression scheme requiring $\tilde O(\min \{ k^4+|T|, k\cdot |T|\})$ space to encode the distances between $S$ and a subset of terminal vertices $T \subseteq V$. This is known as the Okamura-Seymour metric compression problem.

We give an alternative proof of the $x=O(k^3)$ bound that exploits planarity beyond the VC-dimension argument. Namely, our proof relies on cut-cycle duality, as well as on the fact that distances among vertices of $S$ are bounded by $k$. 
Our method implies the following:\\\ (1) An $\tilde{O}(x+k+|T|)$ space compression of the Okamura-Seymour metric, thus improving the compression of Li and Parter to $\tilde O(\min \{k^3+|T|,k \cdot |T| \})$. \\
 (2) An optimal $\tilde{O}(k+|T|)$ space compression of the Okamura-Seymour metric, in the case where the vertices of $T$ induce a connected component in $G$. 
\\
 (3) A tight bound of $x = \Theta(k^2)$ for the family of Halin graphs, whereas the VC-dimension argument is limited to showing $x=O(k^3)$.

\end{abstract}

\section{Introduction}

{\bf Planar metric compression.} The shortest path metric of planar graphs is one of the most popular and well-studied metrics in computer science. The planar graph {\em metric compression} problem is to compactly encode the distances between a subset of $k$ terminal vertices so that we can retrieve the distance between any pair of terminals from the encoding. 
On an $n$-vertex planar graph $G=(V,E)$, a na\"ive encoding uses $\tilde O(\min\{k^2,n\})$ bits (by either storing the $k \times k$ distance matrix or alternatively by storing the entire graph\footnote{Na\"ively, this takes $O(n \log n)$ bits, but can be done with $O(n)$ bits \cite{Turan84,MunroR97,ChiangLL01,BF10}.}). 
It turns out that this na\"ive bound is actually optimal (up to logarithmic factors) for {\em weighted} planar graphs, as shown by Gavoille et al.~\cite{GPPR04}. It is important to note that their lower bound applies even when all terminals lie on a single face. 
The complexity of {\em unweighted undirected} planar graphs is also well-understood. Gavoille et al.~\cite{GPPR04} (see also~\cite{agmw-nocpg-17}) gave a lower bound of $\Omega(\min\{k^2,\sqrt{k\cdot n}\})$, and Abboud et~al.~\cite{agmw-nocpg-17} gave a matching upper bound.

If we are willing to settle for {\em approximate} distances, then there are ingenious compressions requiring only $\tilde{O}(k)$ bits \cite{Thorup04,Klein05,KKS11}. The problem has also been extensively studied (both in the exact~\cite{KZ12,KNZ14,CGH16,ChangGMW18,ChangO20,GHP17} and the approximate~\cite{g-sptmh-01,BG08,CXKR06,KKN15,Che18,Filtser18,FiltserKT19} settings) for the case where we require that the compression is itself a graph (that contains the terminals and preserves their distances). 

\paragraph{The Okamura-Seymour metric compression.}
An important special case for which tight bounds are not yet known, is when the planar graph is unweighted and undirected and we want to encode the $S \times T$ distances between a set of $k$ source terminals $S = \{s_1 , s_2 , \ldots , s_k\}$ lying consecutively on a single face and a subset of target terminal vertices $T \subseteq V$. A query $(v,s_i)$ to the encoding (with $v\in T$ and $s_i \in S$) returns the $v$-to-$s_i$ distance.
\begin{itemize}

\item
When $T = S$, it is possible to exploit the \emph{Unit-Monge} property to obtain an $O(k \log k)$ space encoding with $O(\log k)$ query time~\cite{agmw-nocpg-17}. In fact, even if $T\neq S$ the Unit-Monge property implies an (optimal) $\tilde{O}(|T|+k)$ space encoding, as long as the vertices of $T$ lie (not necessarily consecutively) on single face.

\item
When $T = V$, the MSSP data structure of Eisenstat and Klein~\cite{ek-ltamf-13} gives an $O(n)$ space encoding with $O( \log n )$ query time. 

\item
For arbitrary $S,T$, Li and Parter~\cite{LiParter} recently presented a compression of size $\tilde O(\min \{k^4+|T|, k\cdot |T|\})$ and query time $O(1)$.\footnote{The actual bound stated in~\cite{LiParter} is $\tilde O(k^3 \cdot D+|T|)$ where $D$ is the diameter of the graph. The reason for the additional $D$ factor is that they store all possible distance tuples $d_v = \{d(v,s_i)\}_{i=1}^k$ instead of all possible patterns $p_v$. The reason for the missing $k$ factor is simply a mistake in their paper.}
 This compression is useful algorithmically. In the distributed setting, Li and Parter used it to compute the diameter of a planar graph in $\cOtilde(poly(D))$ rounds where $D$ is the graph's diameter. It was also used to develop an exact distance oracle with subquadratic space and constant query time~\cite{DBLP:journals/corr/abs-2009-14716}.
\end{itemize}

\paragraph{The Li-Parter compression.} 
At the heart of the Li-Parter compression~\cite{LiParter}, is the notion of a \emph{pattern}. Let $d(\cdot,\cdot)$ denote the shortest path metric of $G$. The pattern of a vertex $v \in V$ is the vector 
 $$p_v = \left \langle d(v,s_2) - d(v,s_1) , d(v,s_3) - d(v,s_2) , \ldots , d(v,s_k) - d(v,s_{k-1}) \right \rangle.$$
Since the graph is unweighted, every entry of $p_v$ is in $\{-1,0,1\}$ by the triangle inequality. 
This already gives an efficient way to encode $v$'s distances to $S$: Instead of explicitly storing these distances (using $O(k \log n)$ bits), store $p_v$ and $d(v,s_1)$ (using $O(k + \log n)$ bits). This way, any distance $d(v,s_i)$ can be retrieved by 
$$
d(v,s_i) = d(v,s_1) + \sum_{j=1}^{i-1} p_v[j].
$$

The main contribution of Li and Parter in this context is in showing that, while there are overall $n$ patterns in the graph, there are only $O(k^3)$ {\em distinct} patterns:

\begin{theorem}[\cite{LiParter}]
\label{thm:main-theorem}
The number of distinct patterns over all vertices of the graph is $O(k^3)$.
\end{theorem}

The compression follows easily from the above theorem: Store one table that contains all the distinct patterns of vertices in $T$, and another table that contains for every $v\in T$ the value $d(v,s_1)$ and a pointer to $p_v$ in the first table.
Since there cannot be more than $|T|$ distinct patterns, the size of the first table is $O( \min\{k^4,k\cdot|T|\})$. The size of the second table is $\tilde O(|T|)$. The query time is $O(k)$ but can be improved to $O(1)$ by storing precomputed prefix-sums of every pattern (increasing the size of the first table by a logarithmic factor). 

\paragraph{The original proof of Theorem~\ref{thm:main-theorem}~\cite{LiParter}.} 
Let us assume that the distinguished face is the infinite face. For convenience, we transform\footnote{This transformation was suggested by Li and Parter in their STOC'19 talk.} the problem so that patterns are binary rather than ternary (i.e. over $\{-1,1\}$ instead of $\{-1,0,1\}$). To this end, we subdivide every edge of the graph to get a new (unweighted) graph $G'$. In particular, we replace each edge $\{s_i ,s_{i+1}\}$ of the infinite face with a dummy vertex $w_i$ and edges $\{s_i,w_i\}, \{w_i,s_{i+1}\}$. For every vertex $u$ of $G'$, let $\hat{p}_u$ be the pattern of $u$ w.r.t. the set of vertices $S' = \{s_1,w_1,s_2,w_2, \ldots, s_k,w_k \}$. Observe that the parity of $u$-to-$s_i$ distances is different from the parity of $u$-to-$w_j$ distances, for all $i,j$'s. Hence, $\hat{p}_u$ is a binary vector (i.e. over $\{-1,1\}$). 
Additionally, for every vertex $v$ of $G$ we can retrieve its pattern $p_v$ from $\hat{p}_v$ since $p_v[i] = (\hat{p}_v[2i-1]+\hat{p}_v[2i])/2$. See Figure \ref{binary-transform}. Hence, we henceforth assume that patterns $p_v$ are over $\{-1,1\}$ (i.e. we replace $p_v$ with $\hat{p}_v$). For brevity, we also assume that patterns are of length $k-1$ (rather than $2k-1$).

\begin{figure}[htb]
\centering
\subfloat{\includegraphics[scale=1]{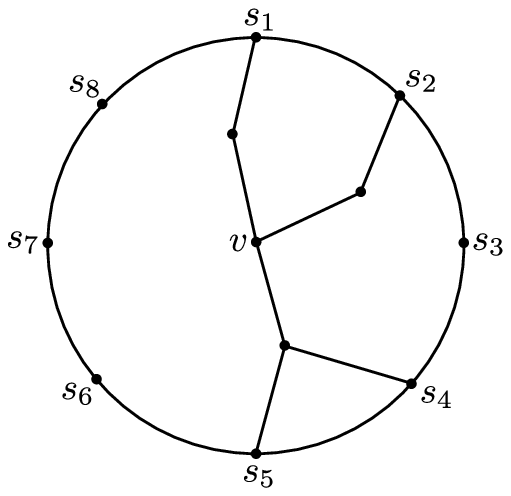}}
\hspace{2cm}
\subfloat{\includegraphics[scale=1]{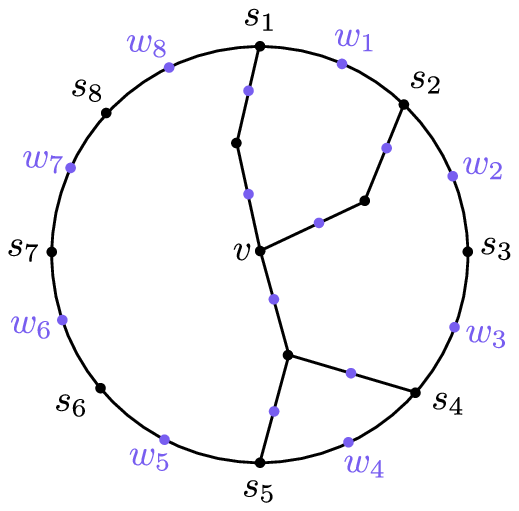}}
\caption{Before (left) and after (right) the transformation that makes all patterns binary. Every edge is subdivided and a new vertex (in color) is put in the middle. In this example, $p_v = \langle 0,1,-1,0,1,1,-1\rangle$ and $\hat{p}_v = \langle 1,-1,1,1,-1,-1,1,-1,1,1,1,1,-1,-1,-1 \rangle$. \label{binary-transform}}
\end{figure}

Li and Parter's VC-dimension argument is based on the simple observation that, by planarity, there cannot be two vertices $v$ and $u$ and 4 indices $a<b<c<d$ such that ${p}_u[a] = -1, {p}_u[b] = 1, {p}_u[c] = -1, {p}_u[d]=1$ but ${p}_v[a] = 1, {p}_v[b] = -1, {p}_v[c] = 1, {p}_v[d]=-1$. The reason is that such $(-1,1,-1,1),(1,-1,1,-1)$ patterns correspond to an illegal configuration of shortest paths in planar graphs.

Consider arranging all the patterns as the rows of a binary matrix $P$. The {\em VC-dimension} $d$ of $P$, is the maximum number of columns in a submatrix of $P$ that contains all possible $2^d$ rows. The above forbidden configuration implies that there is no submatrix with $4$ columns or more that contains all possible rows, hence the VC-dimension of $P$ is at most $3$. By the well known Sauer's Lemma~\cite{Sauer}, this means that there are $O((k-1)^3) = O(k^3)$ distinct rows. This is the entire proof.

\paragraph{Limitations of the original proof.}
It remains an open problem whether the number of distinct patterns in planar graphs is $\Theta(k^3)$ or less (there is a simple $\Omega(k^2)$ lower bound). We do know however that there is no hope of improving $\Theta(k^3)$ using the VC-dimension argument: 
Consider the following set of sequences over $\{-1,1\}^{k-1}$:
$$
\left \{ (-1)^{x_1} \circ 1^{x_2} \circ (-1)^{x_3} \circ 1^{k-1-x_1-x_2-x_3} \; | \; x_1 + x_2 + x_3 < k\right \}
$$
There is no pair of sequences in this set that contains the forbidden $(1,-1,1,-1),(-1,1,-1,1)$ configuration, and yet its cardinality is $\Theta(k^3)$. This means that any improvement to the $O(k^3)$ bound on the number of distinct patterns in planar graphs would have to further exploit structural properties of planar graphs. 
In fact, even in the restricted family of Halin graphs, where we know that there are only $\Theta(k^2)$ distinct patterns (see Section \ref{section:halin}), the VC-dimension argument is limited to proving $O(k^3)$.

\paragraph{Our results and technique.}
We develop a new technique for analyzing and encoding the structure of patterns in a planar graph using {\em bisectors}. The bisector $\beta_i$ associated with vertex $s_i$ is a simple cycle in the dual graph such that all (primal) vertices on the same side of $\beta_i$ have the same $i$'th bit in their patterns. 
We show that any two bisectors are arc-disjoint. This implies the following lemma:

\begin{lemma}\label{lem:adjacentpatterns}
The patterns of every two adjacent vertices in $G$ differ by at most two bits.
\end{lemma}

We then show how to use this property to obtain the following compression (recall that $x$ denotes the number of distinct patterns in $G$):
\begin{theorem}\label{thm:ourcompression}
	There is an $\tilde{O}(x+k+|T|)$ space compression of the Okamura-Seymour metric with $\tilde{O}(n)$ construction time and $\tilde{O}(1)$ query time.  Moreover, for the special case where the vertices of $T$ induce a connected component in $G$, the space is $\tilde{O}(k+|T|)$.  
\end{theorem}
  
By plugging $x=O(k^3)$ from Theorem \ref{thm:main-theorem} (and the trivial compression that stores all $T \times S$ distances) we get an $\tilde{O}(\min\{k^3+|T|,k \cdot |T| \})$ compression (i.e. a factor $k$ improvement over Li and Parter~\cite{LiParter}). 
Moreover, for the special case where the vertices of $T$ induce a connected component in $G$, we obtain an optimal $\tilde{O}(|T|+k)$ space encoding. Recall that, prior to our work, this bound was only known (using the Unit-Monge property) when the vertices of $T$ all lie (not necessarily consecutively) on a single face~\cite{agmw-nocpg-17}. In fact, even in such setting, our method gives $\tilde{O}(|T|+k)$. Thus, our method strictly dominates the one based on Unit-Monge.

An additional benefit of working with bisectors is that they can be used to bound the number $x$ of distinct patterns. 
We show that every two bisectors can cross only $O(k)$ times. Our proof relies not only on the planar structure, but also on the fact that the distance between any two vertices of $S$ is bounded by $k$ (this property is not used in the VC-dimension argument). The set of all bisectors partitions the plane into regions. All (primal) vertices in the same region have the same pattern because they all lie  on the same side of every bisector. Since there are $O(k^2)$ pairs of bisectors, and each pair crosses $O(k)$ times, there are only $O(k^3)$ regions (and hence only $O(k^3)$ distinct patterns). 
This provides an alternative  proof of Theorem \ref{thm:main-theorem}. We believe that our new technique may prove useful in settling the question of the number of distinct patterns in a planar graph. In particular, it may be that a similar argument that uses stronger structural properties will be able to show that the partition induces only $O(k^2)$ regions.
We demonstrate this potential of our technique in Section~\ref{section:halin}, where we show such a bound for a family of graphs that includes Halin graphs:

\begin{theorem}
\label{thm:Halin}
The number of distinct patterns over all vertices of a Halin graph is $O(k^2)$. This bound is tight.
\end{theorem}
In contrast, the VC-dimension argument is limited to proving $O(k^3)$, even on Halin graphs.

\section{Preliminaries}
Let $G=(V,E)$ be an unweighted, undirected planar embedded graph. We prefer to think of $G$ as a directed planar graph with a set of arcs $\mathcal{A}$, such that there is a pair of arcs $uv,vu$ (embedded on the same curve) for every edge $\{u,v\}\in E$. We refer to $u$ and $v$ as the \emph{tail} and \emph{head} of $uv$, respectively. We refer to $uv$ as the \emph{reverse} of $vu$, or simply $rev(vu)$. However, we use the term edge whenever the orientation is not important or when we refer to any of the arcs (possibly both).
We denote by $P_{u,v}$ an arbitrary directed shortest path from $u$ to $v$. For $i<j$ we denote by $S[i,j]$ a path (along the infinite face) $s_i - s_{i+1} - \cdots - s_j$. 
We extend the definition of $rev(\cdot)$ to paths. We denote by $P[w,y]$ the subpath of $P$ between vertices $w$ and $y$. We similarly use $P(w,y]$, $P[w,y)$, and $P(w,y)$ to denote whether the subpath includes the corresponding endpoint(s) or not. We use $\circ$ to denote a concatenation of two paths. Let $C$ be a directed non-crossing cycle in $G$. We denote by ${\sf left}(C)$ and ${\sf right}(C)$ the subgraphs of $G$ that consist of all edges, vertices and faces that are lying to the left and right of $C$, respectively. The arcs of $C$ and their reverses are in both ${\sf left}(C)$ and ${\sf right}(C)$. 

\begin{figure}[htb]
\centering
\includegraphics[scale=0.7]{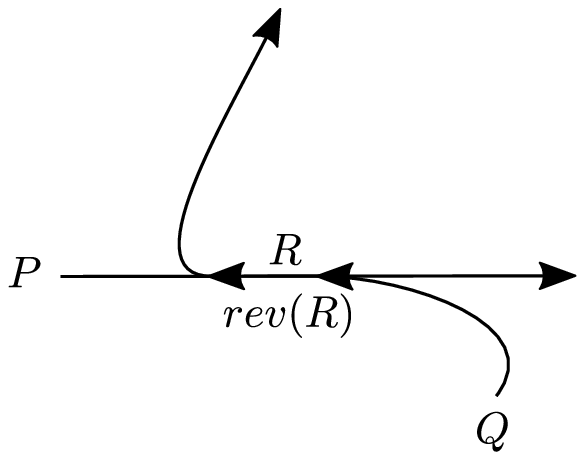}
\caption{An example of two directed paths $P$ and $Q$ that cross at the crossing part $R$.\label{crossing-paths}}
\end{figure}

Let $P$ and $Q$ be directed paths or cycles. We say that they {\em cross at subpath $R$} if, when ignoring their orientation: (1) $R$ is a proper (not a prefix or suffix) subpath of both $P$ and $Q$, and (2) The edges of $Q$ that follow and precede $R$ are in different sides of $P$. See Figure \ref{crossing-paths}. 
We refer to $R$ as a \emph{crossing part} of $P$ and $Q$. 

The dual graph of $G$ is denoted by $G^* = (V^*,E^*)$. Again, we think of $G^*$ as a directed graph with a set of arcs $\mathcal{A}^*$, defined as follows. For every arc $uv \in \mathcal{A}$, there is a corresponding arc $(uv)^* \in \mathcal{A}^*$ such that the tail and head of $(uv)^*$ are the faces that lie to the right and the left of $uv$, respectively. We note that we slightly abuse the notation here, since the dual of $(uv)^*$ is $rev(uv)$ (and not $uv$). For $B \subseteq \mathcal{A}$, let $B^* = \{(uv)^* \; | \; uv \in B \}$. For a cut $X\subseteq V$, let $\delta(X) = \{uv \in \mathcal{A} \; | \; u \in X, v \in V \setminus X \}$. For a cycle $C^*$ in the dual graph we say that $v \in V$ is in ${\sf left}(C^*)$ (resp. ${\sf right}(C^*)$) if the face of $G^*$ that corresponds to $v$ is in ${\sf left}(C^*)$ (resp. ${\sf right}(C^*)$). 

\section{A Bisector-Based Approach to the Okamura-Seymour Compression}
\label{section:alternative-proof}
In this section we present our new proof of Theorem~\ref{thm:main-theorem} and the proofs of Lemma~\ref{lem:adjacentpatterns} and Theorem~\ref{thm:ourcompression}. 
Our main tool is the use of simple dual cycles that we call {\em bisectors}. In Section~\ref{sec:bisect} we define bisectors, and prove that they are arc-disjoint and that this implies Lemma~\ref{lem:adjacentpatterns}. In Section~\ref{sub:intro1} we use it to prove Theorem \ref{thm:ourcompression}. Then, in Section~\ref{sec:patterngraph} we show that the union of all bisectors partitions the graph into regions such that all vertices belonging to the same region have the same pattern. Finally, in Section~\ref{sec:twobisectorscrossktimes} we show that every two bisectors can cross at most $O(k)$ times, implying that the partition induces only $O(k^3)$ regions (and hence only $O(k^3)$ distinct patterns) thus proving Theorem~\ref{thm:main-theorem}. 

\subsection{Bisectors} \label{sec:bisect}

For $1 \leq i \leq k-1$, define the cut $A_i = \{ v \in V \; | \; p_v[i] = -1 \}$. Since we assume the patterns are over $\{-1,1\}$, $V \setminus A_i = \{ v \in V \; | \; p_v[i] = 1\}$. We define the {\em bisector} $\beta_i = \delta(A_i)^*$. Namely, $\beta_i$ consists of all arcs $(uv)^* \in \mathcal{A}^*$ such that $p_u[i]=-1$ and $p_v[i] = 1$. Moreover, every edge $\{u,v\} \in E$ such that $p_u \neq p_v$ belongs to some bisector (possibly more than one). By cut-cycle duality, if the induced subgraphs of $A_i$ and $V \setminus A_i$ are both connected, then $\beta_i$ is a directed simple cycle in the dual graph. The next lemma implies that both induced subgraphs of $A_i$ and $V \setminus A_i$ are connected.

\begin{figure}[htb]
\centering
\includegraphics[scale=0.8]{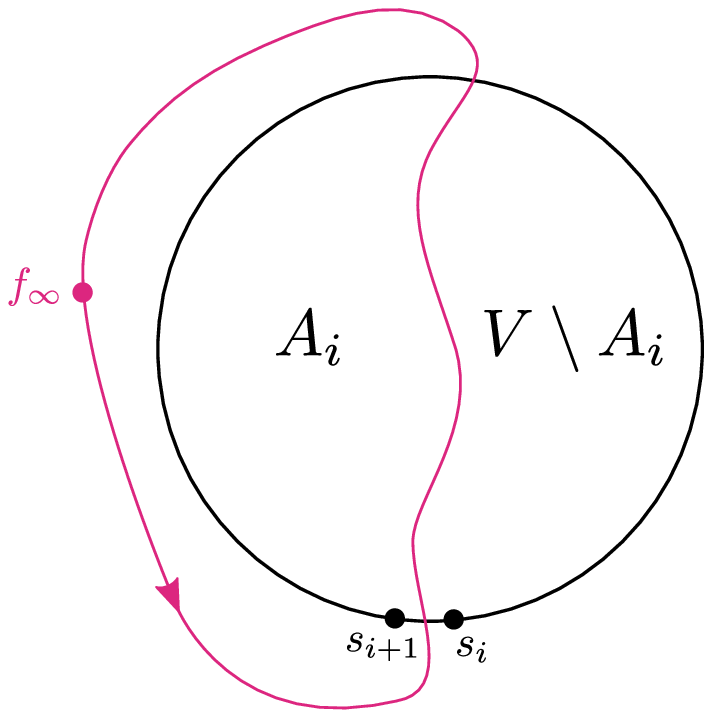}
\caption{The bisector $\beta_i$ and its corresponding cut $A_i$.\label{bisector}}
\end{figure}

\begin{lemma}
\label{lem:paths-contained}
For any $u \in A_i$ (resp. $V \setminus A_i)$, the vertices of $P_{u,s_{i+1}}$ (resp. $P_{u,s_i}$) are in $A_i$ (resp. $V \setminus A_i$).
\end{lemma}
\begin{proof}
Assume that $u \in A_i$ (the proof of the other case is symmetric). Let $v$ be any vertex of $P_{u,s_{i+1}}$ and assume for the sake of contradiction that $v \in V \setminus A_i$. Thus, $d(v,s_{i+1}) = d(v,s_{i})+1$. By the triangle inequality, we get the following contradiction:
\[
d(u,s_i) = d(u,s_{i+1}) + 1 = d(u,v) + d(v,s_{i+1}) +1 = d(u,v) + d(v,s_i) + 2 \geq d(u,s_i) + 2 \ \ \ \ \ \qedhere
\]
\end{proof}

By the above lemma, any two vertices $u,v \in A_i$ (resp. $V \setminus A_i$) are connected in the induced subgraph of $A_i$ (resp. $V \setminus A_i$) by the path $P_{u,s_{i+1}} \circ rev(P_{v,s_{i+1}})$ (resp. $P_{u,s_i} \circ rev(P_{v,s_i})$). This yields the following corollary. 
\begin{corollary}
\label{cor:simple-cycle}
$\beta_i$ is a directed simple cycle in the dual graph.
\end{corollary}
We note that the above corollary implies that for any face $f$, every bisector contains at most two arcs incident to $f$. This shows that there are only $O(k)$ total bit changes between patterns as we go along the vertices of a face $f$. 

Another useful corollary comes from the fact that any edge whose dual is in $\beta_i$ contains endpoints that are both in $A_i$ and $V \setminus A_i$. Therefore:
\begin{corollary}
\label{cor:path-disjoint-bisector}
For any $u \in A_i$ (resp. $u\in V \setminus A_i$), the dual edges of $P_{u,s_{i+1}}$ (resp. $P_{u,s_i}$) are not in $\beta_i$.
\end{corollary}
\noindent
Note that $\beta_i$ has two arcs incident to $f_\infty$, one of them being $(s_{i+1} s_i)^*$. We think of $(s_{i+1} s_i)^*$ as the {\em first} arc of $\beta_i$. See Figure \ref{bisector}. The following lemma shows that bisectors are arc-disjoint. 

\begin{lemma}
\label{lem:no-multiplicities}
Every pair of bisectors $\beta_i, \beta_j$ are arc-disjoint.
\end{lemma}
\begin{proof}
Assume for contradiction that arc $(uv)^*$ appears both in $\beta_i$ and in $\beta_j$. By definition, $u$ belongs to $A_i$ and $A_j$, and $v$ belongs to $V \setminus A_i$ and $V \setminus A_j$.
We first prove that under our assumption, either 
$P_{u,s_{i+1}}$ intersects with $P_{v,s_j}$ or $P_{u,s_{j+1}}$ intersects with $P_{v,s_i}$. To see why, first note that since $P_{v,s_i}$ and $P_{v,s_j}$ are shortest paths, we can choose them to follow a common maximal-length prefix $P_{v,s_i}[v,w] = P_{v,s_j}[v,w]$ for some $w$, and they do not intersect again after $w$. Consider the directed cycle $C = rev(P_{v,s_j}) \circ P_{v,s_i} \circ S[i,j]$ (see Figure~\ref{fig:no-multiplicities}). Notice that by our choice of $P_{v,s_i}$ and $P_{v,s_j}$ and by the fact that $S[i,j]$ lies on the infinite face, $C$ is not necessarily simple but it does not self-cross. We have two cases to consider:

\vspace{0.1in}
\noindent \underline{Case 1}: $u \in {\sf left}(C)\setminus C$. Since $s_{i+1} \in {\sf right}(C)$ then (by the Jordan curve theorem and the fact that all vertices of $S$ lie on the infinite face) $P_{u,s_{i+1}}$ must intersect with $rev(P_{v,s_j}) \circ P_{v,s_i}$. However, by Lemma \ref{lem:paths-contained}, $P_{u,s_{i+1}}$ cannot intersect with $P_{v,s_i}$, therefore it intersects with $rev(P_{v,s_j})$ (and hence with $P_{v,s_j}$).

\vspace{0.1in}
\noindent \underline{Case 2}: $u \in {\sf right}(C)$. Notice that $s_{j+1} \in {\sf left}(C)$. If $s_{j+1}$ is in $P_{v,s_i}$ then $P_{u,s_{j+1}}$ intersects with $P_{v,s_i}$ and we are done. Otherwise, since $s_{j+1}$ is not in $P_{v,s_j}$ by Lemma \ref{lem:paths-contained}, then $s_{j+1} \in {\sf left}(C) \setminus C$. But then again (by the Jordan curve theorem) $rev(P_{u,s_{j+1}})$ (and hence $P_{u,s_{j+1}}$) must intersect with $rev(P_{v,s_j}) \circ P_{v,s_i}$. However, by Lemma \ref{lem:paths-contained}, $P_{u,s_{j+1}}$ cannot intersect with $rev(P_{v,s_j})$, therefore it intersects with $P_{v,s_i}$.
\begin{figure}[htb]
\centering
\subfloat{\includegraphics[scale=1]{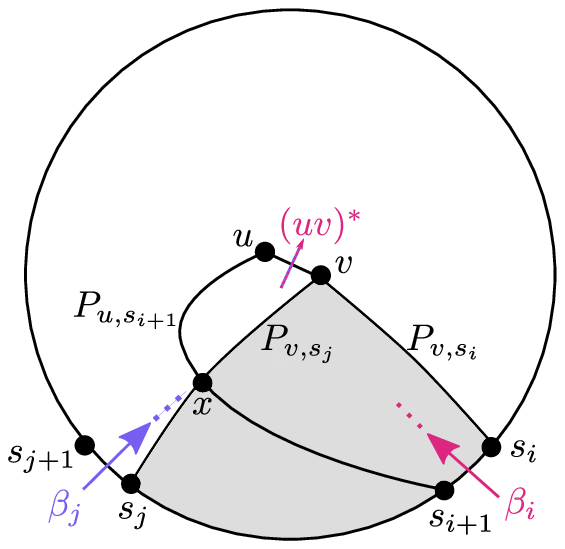}}\ \ \ \ \ \ \ \ \ \ \ \ \ \ \ \ \ \ \ \ 
\subfloat{\includegraphics[scale=1]{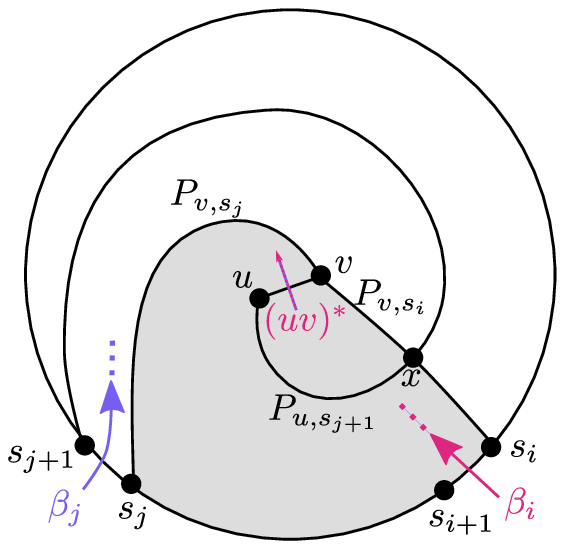}}
\caption{The two cases in the proof of Lemma~\ref{lem:no-multiplicities}. Case 1 on the left. Case 2 on the right. ${\sf right}(C)$ is shaded. For clarity, $C$ is a simple cycle in this example.\label{fig:no-multiplicities}}
\end{figure}

We can therefore continue under the assumption that $P_{u,s_{i+1}}$ and $P_{v,s_j}$ intersect at a vertex $x$ (the other case is symmetric). By the triangle inequality:
\begin{equation*}
\begin{split}
& d(u,s_{j+1}) \leq d(u,x) + d(x,s_j) - 1 \\
& d(u,x) + d(x,s_{i+1}) \leq d(u,v) + d(v,s_i) - 1\\
& d(v,s_i) \leq d(v,x) + d(x,s_{i+1}) - 1 \\
& d(v,x) + d(x,s_j) \leq d(v,u) + d(u,s_{j+1}) -1
\end{split}
\end{equation*}
Since $d(u,v) = d(v,u) = 1$, summing the above inequalities we get the contradiction $0\leq-2$.
\end{proof}

The above lemma shows that two bisectors cannot share an arc.
Note however that it is still possible that a bisector contains {\em reversed} arcs of another bisector. This proves Lemma~\ref{lem:adjacentpatterns}.
Next, we prove Theorem~\ref{thm:ourcompression}.

\subsection{A proof of Theorem \ref{thm:ourcompression}}\label{sub:intro1}
We begin by describing how to compute all the bisectors of the graph and report their arcs in $\tilde{O}(n)$ time. We split every edge $\{s_i , s_{i+1}\}$ by adding a dummy vertex $y_i$ and edges $\{s_i , y_i\},\{y_i,s_{i+1}\}$ of weight $\frac{1}{2}$. Consider a shortest path tree $T_i$ rooted at $y_i$. Notice that the arcs of $\beta_i$ which are not incident to $f_\infty$, are the duals of arcs whose tail is in the subtree rooted at $s_{i+1}$ and head is in the subtree rooted at $s_i$. In the interdigitating tree of $T_i$ (i.e., the tree in the dual graph whose edges are the duals of the edges not in $T_i$), they are precisely the $f_i$-to-$f_\infty$ path without the last arc, where $f_i \neq f_\infty$ is the face incident to $\{s_i,s_{i+1}\}$ in $G$. We can therefore run the MSSP algorithm of Klein~\cite{Klein05} in $\tilde{O}(n)$ time, and report for every $1 \leq i \leq k-1$ all the those arcs of $\beta_i$ in $\tilde O(|\beta_i|)$ time. To report the two arcs of $\beta_i$ which are incident to $f_\infty$, one of them is trivially $(s_{i+1} s_i)^*$ and the other one is determined by last arc of the above $f_i$-to-$f_\infty$ path. Since by Lemma~\ref{lem:no-multiplicities} the arcs of bisectors are disjoint, this takes time $\tilde{O}(n + \sum_{i=1}^{k-1} |\beta_i|) = \tilde O(n)$ time. In particular, we can label every edge $\{u,v\} \in E$ by the (at most two) bisectors that use $(uv)^*$ and $(vu)^*$. I.e., the bits that change between $p_u$ and $p_v$.

We next describe the compression scheme. Recall that, by storing $d(v,s_1)$ for every $v \in T$, a query $d(v,s_i)$ (with $v\in T$ and $s_i \in S$) boils down to extracting $p_v$ and computing its $(i-1)$'th prefix-sum.
Let $\mathcal{T}$ be a spanning tree of $G$. Label each edge $\{u,v\}$ of $\mathcal{T}$ by the (at most two) bits that change between the patterns of $u$ and $v$. Note that there could be many (potentially $\Omega(n)$) nodes of $\mathcal{T}$ that correspond to the same pattern. 
In order to decrease the size of $\mathcal{T}$ to be $x$ (the number of distinct patterns in $G$), we root $\mathcal{T}$ at some arbitrary node $u$. Then, for every two nodes $v,w$ of $\mathcal{T}$ s.t $p_v=p_w$ (and w.l.o.g. $v$ is not a descendent of $w$) we remove the node $w$ and turn all it's children to be children of $v$ (their edge labels remain the same).
We repeat this process until the size of the tree is $x$. We denote the resulting tree by $\mathcal{T}'$.
 Let $Q$ be an Euler-tour of $\mathcal{T}'$ starting from the root $u$. Consider the patterns of the nodes as we go along $Q$, starting from $p_{u}$. In each step, the pattern only changes in at most two bits (according to the edge labels). Therefore, we can maintain all these $O(|Q|)= O(x)$ versions of the pattern using a \emph{persistent}~\cite{driscoll1989making} data structure for prefix-sum (e.g., using persistent segment trees~\cite{segmenttree}). Such a data structure supports both updates and prefix-sum queries to any version in $\tilde O(1)$ time and uses $\tilde O(|Q| + k) = \tilde O(x+k)$ space. Finally, for every vertex $v \in T$ let $q_v$ be a node in $Q$ whose corresponding pattern is $p_v$. We store a pointer from $v$ to the version of the persistent data structure at $q_v$, using additional $\tilde{O}(|T|)$ bits overall.

We now give a randomized $\tilde{O}(n)$ time algorithm for constructing $\mathcal{T'}$ (and hence the compression). An arbitrary spanning tree $\mathcal{T}$ can be computed in $O(n)$ time. Assume that every edge $\{w,v\}$ of $\mathcal{T}$ is labeled by the (at most two) bits that change between $p_w$ and $p_v$. Let us compute the pattern of the root $u$ of $\mathcal{T}$ with a single-source shortest-paths computation in $G$. We also compute the \emph{Karp-Rabin fingerprint}~\cite{KarpRabin} $\fp(p_u)$ of $p_u$. Such fingerprints are appealing because: (1) for any $p_w \neq p_v$, we have that $\fp(p_w) \neq \fp(p_v)$ with high probability, and (2) given $\fp(S_1)$ and $\fp(S_2)$ of two strings $S_1,S_2$ we can compute in $O(1)$ time the fingerprint of the concatenation $\fp(S_1\circ S_2)$. Thus, if we maintain a complete binary tree  on top of the pattern where each node contains the fingerprint of its subtree (and in particular, the root contains the fingerprint of the entire pattern), then we can update this tree in $O(\log k)$ time after changing one or two bits in the pattern. 
 
We maintain the fingerprints in a dictionary initially containing only $\fp(p_u)$.
We process the nodes of $\mathcal{T}$ starting from $u$, maintaining a queue of next-to-visit nodes. When we process a node $v$, we compute $\fp(p_v)$ from the fingerprint of $v$'s parent, by flipping the bits according to the edge label (in $O(\log k)$ time). We then try to add $\fp(p_v)$ to the dictionary. If we find a collision with some vertex $w$ (namely, $\fp(p_v) = \fp(p_w)$) then we delete $v$ from $\mathcal{T}$, and set the children of $v$ to be children of $w$ in $\mathcal{T}$. In any case, we add the children of $v$ to the queue so they will be processed later. Notice that a node is visited only after all its ancestors have been visited. Therefore, we can always compute its fingerprint and we never move children from a vertex to its descendent, so $\mathcal{T}$ remains a tree. In addition, the parent of every node changes or gets deleted at most once, hence the running time is $\tilde{O}(n)$. Overall, in $\tilde{O}(n)$ time we construct $\mathcal{T'}$ and the dictionary (both of size $x$).  

\begin{itemize}
	\item In the special case where the vertices of $T$ induce a connected component in $G$, we can skip the first part of the algorithm and simply take a path $Q$ that traverses only the vertices of $T$. The rest of the construction remains the same and since $|Q| = O(|T|)$, the size of the compression is $\tilde{O}(|T|+k)$. 
	\item In the special case where the vertices of $T$ all lie on a single face (but not necessarily consecutively), let $Q$ be a path that visits all the vertices of the face in clockwise order. By Corollary~\ref{cor:simple-cycle}, the total number of bit changes between patterns of consecutive vertices along $Q$ is $O(k)$. Therefore, the number of patterns encountered is $O(k)$ and hence we get an $\tilde{O}(|T|+k)$ compression for this case as well.

\end{itemize}
This completes the proof of Theorem~\ref{thm:ourcompression}.

\subsection{The bisector graph and the pattern graph}\label{sec:patterngraph}

The \emph{bisector graph} $G_\mathcal{B}$ is the subgraph of $G^*$ composed of the union of all the bisectors. The faces of $G_\mathcal{B}$ represent the patterns of $G$ in the following way. 
\begin{lemma}
\label{lem:patterns-of-faces}
For every $u,v \in V$, if $u$ and $v$ are embedded inside the same face $f$ of $G_\mathcal{B}$, then $p_u=p_v$.
\end{lemma}

\begin{proof}
Notice that $G_\mathcal{B}$ is a connected graph because all the bisectors are incident to $f_\infty$. Hence, $f$ is a simple cycle in $G^*$. Let $G_f$ be the subgraph of $G$ embedded inside the face $f$. Since there are no bisector edges embedded inside $f$, then in $G_f$ there is no pair of adjacent vertices that have different patterns. Since $f$ is a simple cycle, then by cut-cycle duality $G_f$ is a connected subgraph. Therefore, there exists a $u$-to-$v$ path in $G_f$, and every pair of adjacent vertices in this path have the same pattern. Hence $p_u=p_v$.
\end{proof}

By the above lemma, every pattern $p$ of $G$ corresponds to a unique nonempty subset of faces of $G_\mathcal{B}$. More precisely, a pattern $p$ corresponds to all the faces of $G_\mathcal{B}$ such that the vertices of $G$ embedded in these faces have pattern $p$. In particular, the number of faces of $G_\mathcal{B}$ is an upper bound on the number of distinct patterns in $G$. Therefore, if we could prove that $G_\mathcal{B}$ has $O(k^3)$ faces we would be done. Unfortunately, this is not the case. There can be as many as $\Omega(n)$ faces of $G_\mathcal{B}$ that correspond to the same pattern (see Figure \ref{repeating-faces}). To tackle this, we transform $G_\mathcal{B}$ into a new graph $G_\mathcal{P}$ (called the \emph{pattern graph}) that has only $O(k^3)$ faces and whose faces still represent all the distinct patterns of $G$.

\begin{figure}[htb]
\centering
\includegraphics[scale=1]{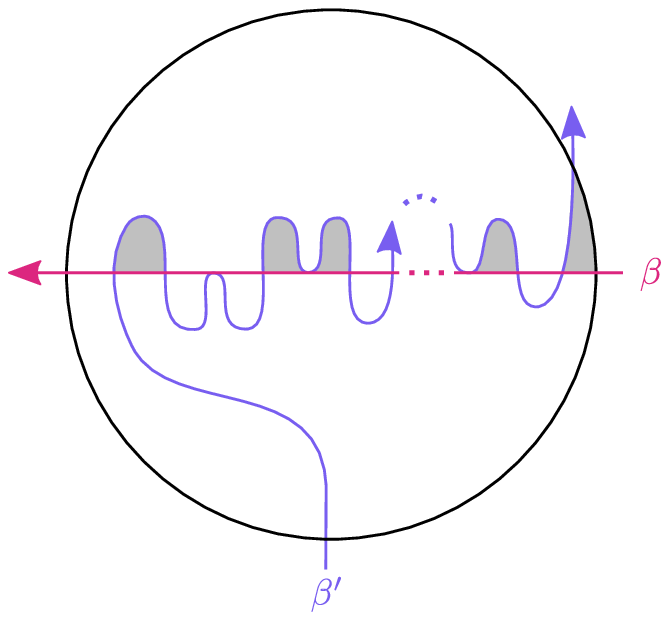}
\caption{The shaded faces all correspond to the same pattern (assuming no other bisector crosses or separates them). They are formed when the two bisectors either cross each other or just touch (intersect without crossing). 
We will later see that two bisectors can cross each other at most $O(k)$ times, but, they can touch $\Omega(n)$ times, creating $\Omega(n)$ faces that correspond to the same pattern.}
\label{repeating-faces}
\end{figure}

The {\em pattern graph} $G_\mathcal{P}$ is obtained by applying on $G_\mathcal{B}$ the following two-phase procedure:\\ (1) A \emph{Peel} phase: 
Recall that while Lemma \ref{lem:no-multiplicities} says that every two bisectors are arc-disjoint, it is still possible that one bisector contains {\em reversed} arcs of another. In the peel phase, we re-embed the bisectors so that no bisector contains reversed arcs of another bisector. After the peel phase, crossings and touchings occur only at vertices (rather than subpaths).
\\(2) A \emph{Merge} phase: In the merge phase, we merge faces that correspond to the same pattern and share a common vertex.

\paragraph{Peel Phase.} For every two bisectors $\beta$ and $\beta'$, consider the set of maximal-length subpaths $R$, such that $R$ is a subpath of $\beta$ and $rev(R)$ is a subpath of $\beta'$. If the arc of $\beta$ that follows $R$ is in ${\sf right}(\beta')$ (resp. ${\sf left}(\beta')$), then we re-embed every arc of $R$ on a new curve lying to the right (resp. left) of its reverse. 
See Figure \ref{peel-phase}. Note that the peel phase does not create any new crossings between $\beta$ and $\beta'$. 

\begin{figure}[htb]
\begin{minipage}{0.5\textwidth}
 \hspace{1.5cm}
\includegraphics[scale=1]{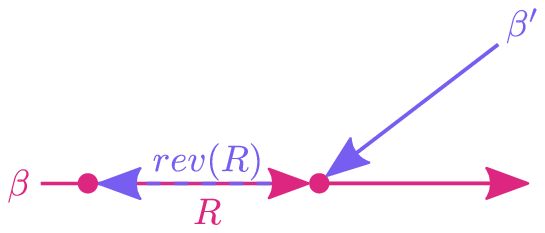}
\end{minipage}
\begin{minipage}{0.5\textwidth}
 \vspace{0.45cm}
 \hspace{0.7cm}\includegraphics[scale=1]{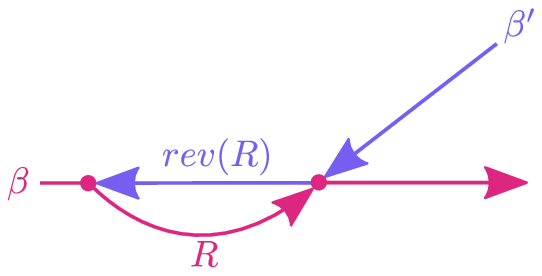}
\end{minipage}
\caption{Before (left) and after (right) a peel phase. In this example, $R$ is a single arc, and the arc of $\beta$ that follows $R$ is in ${\sf left}(\beta')$.\label{peel-phase}}
\end{figure}

\paragraph{Merge Phase.} For every vertex $g \neq f_\infty$ of a bisector $\beta$, if any other bisector crosses $\beta$ at $g$ then we do nothing. Otherwise, we split $g$ into two copies. All the arcs in ${\sf left}(\beta)$ that are incident to $g$ are connected to one copy, and all the arcs in ${\sf right}(\beta)$ that are incident to $g$ are connected to the other copy. Finally, we replace the arcs of $\beta$ that are incident to $g$ (say $fg$ and $gh$) by a single arc $fh$. See Figure \ref{merge-phase}. Note that if $g$ is not incident to any bisector other than $\beta$, then the merge phase simply contracts the arc $gh$. We repeat this process until there are no such bisector pairs in the graph.

\begin{figure}[htb]
\begin{minipage}{0.5\textwidth}
 \hspace{2cm}
\includegraphics[scale=1]{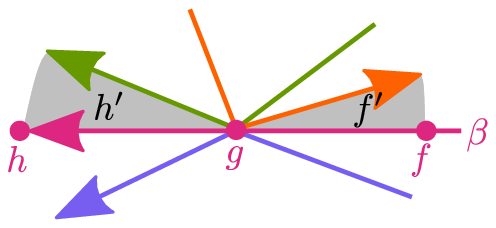}
\end{minipage}
\begin{minipage}{0.5\textwidth}
 \hspace{1.5cm}\includegraphics[scale=1]{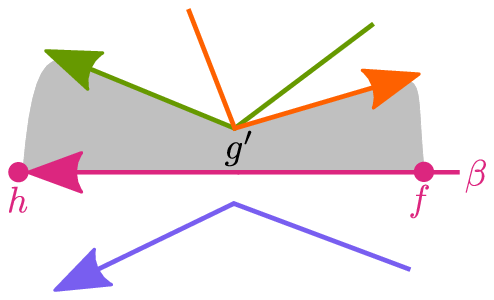}
\end{minipage}
\caption{Before (left) and after (right) a merge phase. The (shaded) face $g'$ is obtained by merging faces $h'$ and $f'$.\label{merge-phase}}
\end{figure}
We now show that the above two-phase procedure maintains the relation between patterns in $G$ and faces in $G_\mathcal{P}$. Namely, that every pattern in $G$ corresponds to a unique nonempty subset of faces of $G_\mathcal{P}$. To this end, we extend the definition of patterns to faces of $G_\mathcal{B}$. This step is necessary since the peel phase creates faces that do not correspond to primal vertices.

We define the \emph{pattern} of a face $f$ of $G_\mathcal{B}$, denoted $p_f$, to be the length $k-1$ vector where $p_f[i] = -1$ (resp. $p_f[i] = 1$) if $f$ is a face in ${\sf left}(\beta_{i+1})$ (resp. ${\sf right}(\beta_{i+1})$) in $G_\mathcal{B}$. The definition remains the same for any graph we obtain from $G_\mathcal{B}$ during the two-phase procedure. The following two propositions show that this definition is consistent with the original definition of patterns (of vertices).
\begin{proposition}
\label{prop:face-patterns}
Let $v \in V$ be a vertex embedded inside a face $f$ of $G_\mathcal{B}$. Then $p_v = p_f$.
\end{proposition}
\begin{proof}
Let $0\leq i \leq k-2$. If $p_f[i]=1$ then $f$ is in ${\sf right}(\beta_{i+1})$ by definition. Since $G_\mathcal{B}$ is a subgraph of $G^*$, then in $G^*$ $v$ is also embedded in ${\sf right}(\beta_{i+1})$. Hence, $v \in V \setminus A_i$ and therefore $p_v[i] = 1$. A symmetric argument shows that if $p_f[i]=-1$ then $p_v[i]=-1$.
\end{proof}
By Proposition~\ref{prop:face-patterns}, all the faces of $G_\mathcal{B}$ that correspond to a pattern $p_v$ have the same face pattern. Notice that the peel phase does not change the patterns of existing faces. It can only add new faces to the graph, but no vertex of $G$ is embedded in any of these new faces. 
Hence, the relation is preserved after the peel phase. Next we show that after a merge step, every pattern still corresponds to a unique subset of faces (i.e., we show that we do not merge faces that corresponded to different patterns). Consider a single merge step happening at $g$ (as illustrated in Figure \ref{merge-phase}). Denote by $f'$ (resp. $h'$) the face lying to the left of $fg$ (resp. $gh$). Namely, $f'$ and $h'$ are the faces that get merged (a symmetric argument holds when they lie to the right of $fg$ and $gh$). Let $g'$ denote the face obtained by merging $f'$ and $h'$.
\begin{proposition}
\label{prop:merge-phase}
$p_{f'} = p_{h'} = p_{g'}$.
\end{proposition}
\begin{proof}
Since no bisector crosses $\beta$ at $g$, then $fg$ and $gh$ belong to the same side of every bisector. This, together with the fact that $fg$, $gh$, and their reverses do not belong to any other bisector, implies that $f'$ and $h'$ also belong to the same side of every bisector. Hence $p_{f'}=p_{h'}$. Now consider the arc $fh$ after the merge. 
Since $f$ and $h$ belong to the same side of every bisector as $fg$ (and $gh$), then $g'$ also belongs to the same side of every bisector, hence $p_{f'} = p_{h'} = p_{g'}$. 
\end{proof}
By proposition \ref{prop:face-patterns}, if $f'$ or $h'$ are faces that correspond to $p_v$ then they do not correspond to any $p_u \neq p_v$. By Proposition \ref{prop:merge-phase}, we can set $g'$ to correspond to $p_v$, and the set of faces corresponding to every pattern remains unique. This yields the following corollary.
\begin{corollary}
\label{cor:pattern-graph}
Every pattern of $G$ corresponds to a unique subset of faces of $G_\mathcal{P}$.
\end{corollary}
Finally, we show that the number of faces in $G_\mathcal{P}$ depends linearly on the number of bisector crossings. 
Let $t$ be the total number of bisector crossings in $G_\mathcal{P}$. That is, $t$ is the sum of the number of crossings between all pairs of bisectors.
\begin{lemma}
\label{lem:faces-bounded-by-crossings}
The number of faces in $G_\mathcal{P}$ is $O(t+k)$.
\end{lemma}
\begin{proof}
By Euler's formula, it suffices to show that the number of arcs in $G_\mathcal{P}$ is $O(t+k)$. For every arc $fg$ in $G_\mathcal{P}$, where neither $f$ nor $g$ is $f_\infty$, the arc $fg$ belongs to some bisector $\beta$. Moreover, there must exist some other bisector that crosses $\beta$ at $f$. Otherwise, the arc would have been removed in the merge phase. Consider all the bisectors that cross $\beta$ at $f$ in a clockwise order around $f$ starting at $fg$. Let $\beta'$ be the one following $fg$. Then we charge $fg$ to the crossing of $\beta$ and $\beta'$ at $f$. Notice that at most two arcs will be charged to this crossing of $\beta$ and $\beta'$ (the arc $fg$ and the arc of $\beta'$ whose tail is $f$). 
Overall, we have charged $O(t)$ arcs. 
The only arcs that did not get charged are the $2(k-1)$ arcs incident to $f_\infty$. Therefore, the number of arcs in $G_\mathcal{P}$ is $O(t+k)$.
\end{proof}

In the next subsection, we prove that every pair of bisectors can cross at most $O(k)$ times. Since there are $O(k^2)$ pairs of bisectors, the total number of crossings is then $t=O(k^3)$, which by Corollary \ref{cor:pattern-graph} and Lemma \ref{lem:faces-bounded-by-crossings} implies Theorem \ref{thm:main-theorem}.

\subsection{Two bisectors can cross only $O(k)$ times}\label{sec:twobisectorscrossktimes}

Let $\beta_i$ and $\beta_j$ be two bisectors in $G^*$ that cross each other at least once. Let $R_1 , R_2 , \ldots R_r$ be their crossing parts that do not contain $f_\infty$, sorted by their order of appearance along $\beta_i$. We note that since $\beta_i$ and $\beta_j$ are simple cycles, the crossing parts must be disjoint. In Lemma \ref{lem:opposite-orientation} we show that the crossing parts appear in reverse order along $\beta_j$, and in Lemma \ref{lem:bounded-crossing} we use this fact to prove that the number of crossings $r$ is at most $O(k)$. 
We begin by defining an important configuration of bisectors and shortest paths. 
\begin{figure}[htb]
\centering
\includegraphics[scale=1]{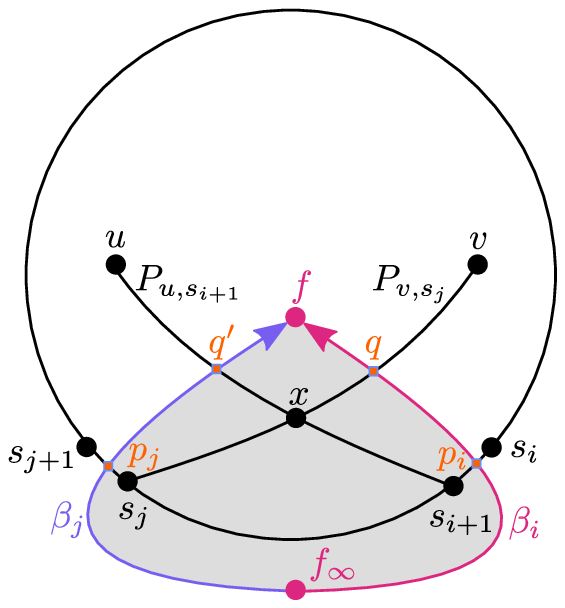}
\caption{A configuration of bisectors $\beta_i,\beta_j$ and shortest paths $P_{u,s_{i+1}},P_{v,s_j}$ that must intersect. ${\sf right}(C^*)$ is shaded. \label{crossing-config}}
\end{figure}

\begin{lemma}
\label{lem:crossing-config}
Let $C^* = \beta_j[f_\infty,f] \circ rev(\beta_i[f_\infty,f])$ be a simple cycle, and let $u \in A_i ,v \in V \setminus A_j$ (resp. $u \in V \setminus A_i, v \in A_j$) be two vertices in ${\sf left}(C^*)$ (resp. ${\sf right}(C^*)$). Then $P_{u,s_{i+1}}$ and $P_{v,s_j}$ (resp. $P_{u,s_i}$ and $P_{v,s_{j+1}}$) must intersect at some vertex $x$.
\end{lemma}
\begin{proof}
We focus on the case where $u \in A_i$ and $v \in V \setminus A_j$ are vertices in ${\sf left}(C^*)$ (the proof of the other case is symmetric). See Figure \ref{crossing-config}.  
We assume that $G$ and $G^*$ are embedded on the same surface, such that for every $wy \in \mathcal{A}$, the curves of $wy$ and $(wy)^*$ intersect in their middles at a single \emph{point} $p$ on the surface. See Figure \ref{dual-primal-point}. 
\begin{figure}[htb]
\centering
\includegraphics[scale=0.9]{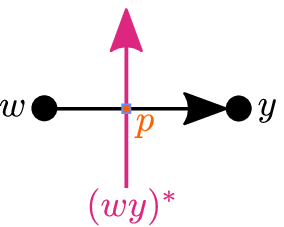}
\caption{An arc $wy \in \mathcal{A}$ that intersects with its dual $(wy)^* \in \mathcal{A}^*$ at a middle point $p$.\label{dual-primal-point}}
\end{figure}
We refer to the two parts of the curve of $uv$ as $(u-p)$ and $(p-v)$. For a path $Q$ that contains arc $uv$, we slightly abuse notation and use $Q[\cdot,p]$ and $Q[p,\cdot]$ to denote a prefix and suffix of the curve of $Q$. In addition, we say that a path $P$ of the primal graph crosses a path $Q$ of the dual graph, if $P$ contains an arc $uv$ whose dual or reversed dual is in $Q$. In particular, it means that there exist a common point $p$ (in the middle of $uv$), such that $(u-p)$ and $(p-v)$ are on different sides of $Q$.
Let $p_j$ be the point in the middle of $s_j s_{j+1}$, and let $p_i$ be the point in the middle of $s_i s_{i+1}$. 

Notice that $s_{i+1},s_j \in {\sf right}(C^*) \setminus C^*$ by definition, and that $v \in {\sf left}(C^*) \setminus C^*$ by assumption (and the fact that $v$ is not part of $C^*$ because $v$ is a primal vertex). Hence, $P_{v,s_j}$ must cross $C^*$. However, by Corollary \ref{cor:path-disjoint-bisector} it cannot cross $\beta_j[f_\infty,f]$, hence it must cross $\beta_i[f_\infty,f]$. This means that there is a point $q$ that is common to both $P_{v,s_j}$ and $\beta_i[f_\infty,f]$. In particular, let $q$ be the last point along the curve of $P_{v,s_j}$ that is also along the curve of $C^*$. Notice that $P_{v,s_j}[q,s_j] \circ (s_j- p_j) $ is a chord inside the cycle $C^*$. Similarly, there exists a point $q'$ along the curve of $P_{u,s_{i+1}}$ such that $P_{u,s_{i+1}}[q',s_{i+1}] \circ (s_{i+1} - p_i)$ is a chord in $C^*$. The endpoints of the chords appear in clockwise order along $C^*$ as $(q,p_i,p_j,q')$. 
It is well known that two chords in such a configuration must intersect. Therefore, there exist a primal vertex $x$ that is common to both $P_{v,s_j}[q,s_j]$ and $P_{u,s_{i+1}}[q',s_{i+1}]$. 
\end{proof}

It is important to remark that Lemma \ref{lem:crossing-config} holds even when the cycle $C^*$ is a non-self-crossing non-simple cycle. Namely, if $C^*$ intersects with itself, we let $g$ be the first intersection vertex in $\beta_i[f_\infty,f]$, and define a cycle $\hat{C}^* = \beta_j[f_\infty,g] \circ rev(\beta_i[f_\infty,g])$. Note that since $\beta_i$ and $\beta_j$ are simple cycles, and by our choice of $g$, there are no intersections in $\hat{C}^*$. Clearly, we also have that $v,u \in {\sf left}(\hat{C}^*)$. Thus, we apply the Lemma to $\hat{C}^*$ instead of $C^*$.

\begin{corollary}
\label{cor:path-cross-bisector}
$P_{v,s_j}[v,x]$ (resp. $P_{u,s_{i+1}}[u,x]$) contains an edge whose dual is in $\beta_i$ (resp. $\beta_j$).
\end{corollary}

We are now in the position to prove the two main lemmas of this section. 

\begin{lemma}
\label{lem:opposite-orientation}
Let $\beta_i$ and $\beta_j$ be two crossing bisectors. Let $R_1 , R_2 , \ldots R_r$ be their crossing parts along $\beta_i$. Then, the crossing parts along $\beta_j$ are reversed $R_r , R_{r-1} , \ldots R_{1}$.
\end{lemma}
\begin{proof}

We assume that $i<j$. For $\ell \leq r$, we say that $\beta_j$ enters $R_\ell$ from ${\sf left}(\beta_i)$ (resp. ${\sf right}(\beta_i)$) if the arc of $\beta_j$ that precedes $R_\ell$ is in ${\sf left}(\beta_i)$ (resp. ${\sf right}(\beta_i)$). For brevity, we assume that every $R_\ell$ is a single vertex in $V^*$. Finally, we assume without loss of generality that $s_j$ is in ${\sf left}(\beta_i)$ (the proof of the other case is symmetric).

Assume for the sake of contradiction that the order of appearance is not the reverse order. 
Then there exists a pair of crossing parts $R_a$ and $R_b$ such that: (1) $a < b$, (2) $R_b$ is the crossing part following $R_a$ in $\beta_j$, and (3) $a$ is minimal among such pairs. Consider the cycle $C^* = \beta_j[f_\infty,R_a] \circ rev(\beta_i[f_\infty,R_a])$. Note that $C^*$ is non-crossing since if there exists $R_c$ for $c<a$ that $\beta_j[f_\infty,R_a]$ crosses, then $a$ wouldn't be minimal. We have two cases to consider: 

\begin{figure}[htb]
\centering
\subfloat{\includegraphics[scale=1]{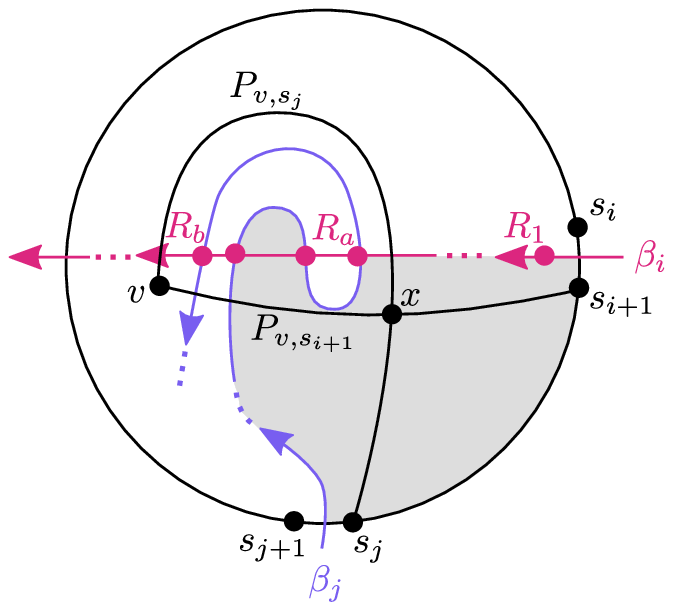}}
\hspace{1cm}
\caption{Case 1 in the proof of Lemma \ref{lem:opposite-orientation}. The dotted parts represent parts in which there may be crossings. ${\sf right}(C^*)$ is shaded. In this example, $P$ is crossed by $\beta_j[f_\infty,R_a]$ exactly twice. \label{op-proof-case1}}
\end{figure}

\noindent \underline{Case 1}: $\beta_j$ enters $R_a$ from ${\sf left}(\beta_i)$ (hence it enters $R_b$ from ${\sf right}(\beta_i)$). Let $(uv)^*$ be the arc of $\beta_j$ that follows $R_b$. We next show that $v$ and $C^*$ form the configuration of Lemma \ref{lem:crossing-config}, in the special case of $v=u$. 

First we show that $v \in V \setminus A_j$ and $v \in A_i$. Note that $v \in V \setminus A_j$ by definition, since $v$ is the primal vertex lying to the right of an arc of $\beta_j$. Since $(uv)^*$ follows the crossing part $R_b$, and since $\beta_j$ enters $R_b$ from ${\sf right}(\beta_i)$, then $(uv)^* \in {\sf left}(\beta_i) \setminus \beta_i$. Hence, $v \in {\sf left}(\beta_i)$ and therefore $v \in A_i$.

Next we show that $v \in {\sf left}(C^*)$. For this, we show that $R_b \in {\sf left}(C^*) \setminus C^*$, hence the arcs incident to $R_b$ and their corresponding primal vertices are in ${\sf left}(C^*)$. Consider the path $P = \beta_i[R_a, R_b]$. The first arc of $P$ is in ${\sf left}(C^*)$ by the case assumption. To show that $R_b \in {\sf left}(C^*) \setminus C^*$ we will show that $P$ crosses the cycle $C^*$ an even number of times. Notice that $P$ does not cross $\beta_i[f_\infty,R_a]$ since $\beta_i$ is a simple cycle. It therefore remains to show that $P$ crosses $\beta_j[f_\infty,R_a]$ an even number of times (equivalently, we show that $\beta_j[f_\infty,R_a]$ crosses $P$ an even number of times). To this end, let us define a cycle $\hat{C}^* = P \circ rev(\beta_j[R_b,R_a])$. Notice that $\hat{C}^*$ is non-crossing since $R_b$ follows $R_a$ in $\beta_j$. 
Notice that the first and last arcs of $\beta_j[f_\infty,R_a]$ are both in ${\sf left}(\beta_i)$. Also notice that ${\sf right}(\hat{C}^*)$ is included in ${\sf right}(\beta_i)$ (by the case assumption), therefore the first and last arcs of $\beta_j[f_\infty,R_a]$ are in ${\sf left}(\hat{C}^*)$. Hence, $\beta_j[f_\infty,R_a]$ crosses $\hat{C}^*$ an even number of times. Since $\beta_j[f_\infty,R_a]$ can only cross $\hat{C}^*$ at $P$ (as $\beta_j$ is a simple cycle), it must cross $P$ an even number of times.

We can thus apply Lemma \ref{lem:crossing-config}, and conclude that $P_{v,s_j}$ and $P_{v,s_{i+1}}$ intersect at vertex $x$, and that $P_{v,s_{i+1}}[v,x]$ contains an edge whose dual is in $\beta_j$ (by Corollary \ref{cor:path-cross-bisector}). Since the lengths of $P_{v,s_j}[v,x]$ and $P_{v,s_{i+1}}[v,x]$ are the same, $P_{v,s_{i+1}}[v,x] \circ P_{v,s_j}[x,s_j]$ is a shortest $v$-to-$s_j$ path. However, since $P_{v,s_{i+1}}[v,x]$ contains an edge of $\beta_j$, we get a contradiction to Corollary \ref{cor:path-disjoint-bisector}.

\vspace{0.1in}

\begin{figure}[htb]
\centering
\subfloat{\includegraphics[scale=1]{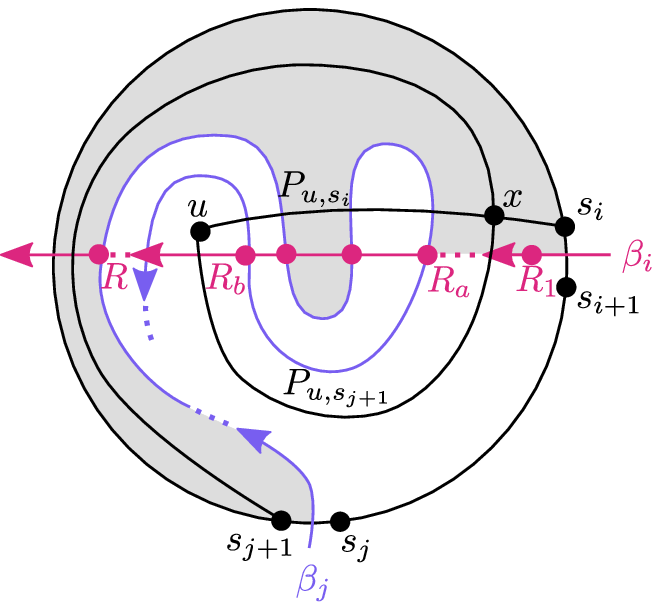}}
\hspace{1cm}
\caption{Case 2 in the proof of Lemma \ref{lem:opposite-orientation}. The dotted parts represent parts in which there may be crossings. ${\sf left}(C^*)$ is shaded. In this example, $P$ is crossed by $\beta_j[f_\infty,R_a]$ exactly twice.
\label{op-proof-case2}}
\end{figure}

\noindent \underline{Case 2}: 
$\beta_j$ enters $R_a$ from ${\sf right}(\beta_i)$ (hence it enters $R_b$ from ${\sf left}(\beta_i)$. Let $(uv)^*$ be the arc of $\beta_j$ that follows $R_b$. We next show that $u$ and $C^*$ form the configuration of Lemma \ref{lem:crossing-config} (again, in the special case of $u=v$). By symmetric arguments to Case 1, we can show that $u \in A_j$ and $u \in V \setminus A_i$. We therefore only need to show that $u \in {\sf right}(C^*)$. 

To show that $u \in {\sf right}(C^*)$, it suffices to show that $R_b \in {\sf right}(C^*) \setminus C^*$ (hence the arcs incident to $R_b$ and their corresponding primal vertices are in ${\sf right}(C^*)$). 
Consider the path $P=\beta_i[R_a,R_b]$. The first arc of $P$ is in ${\sf right}(C^*)$ by the case assumption. 
To show that $R_b \in {\sf right}(C^*) \setminus C^*$, we will show that $P$ crosses the cycle $C^*$ an even number of times. Since $P$ does not cross $\beta_i[f_\infty,R_a]$, we need to show that $\beta_j[f_\infty,R_a]$ crosses $P$ an even number of times (here is where the argument will differ from Case 1). 

Let $\hat{C}^* = P \circ rev(\beta_j[R_b,R_a])$ be a non-crossing cycle. 
Let $R$ be the first crossing between $\beta_j[f_\infty,R_a]$ and $\beta_i$. 
$R$ exists and is not $R_a$ by the assumption that $s_j \in {\sf left}(\beta_i)$, and by the case assumption. 
Note that $\beta_j[f_\infty,R]$ does not cross $\beta_i$ (and hence does not cross $P$) by definition of $R$, so it remains to show that the number of crossings between $\beta_j[R,R_a]$ and $P$ is even.
Notice that, since $s_j \in {\sf left}(\beta_i)$, the first and last arcs of $\beta_j[R,R_a]$ are both in ${\sf right}(\beta_i)$. Also note that ${\sf left}(\hat{C}^*)$ is contained in ${\sf left}(\beta_i)$. Therefore, the first and last arcs of $\beta_j[R,R_a]$ are in ${\sf right}(\hat{C}^*)$. Hence, $\beta_j[R,R_a]$ crosses $\hat{C}^*$ an even number of times. Since it can only cross $\hat{C}^*$ at $P$, then $\beta_j[R,R_a]$ crosses $P$ an even number of times.

We can thus apply Lemma \ref{lem:crossing-config}, and conclude that $P_{u,s_{j+1}}$ and $P_{u,s_i}$ intersect at vertex $x$, and that $P_{u,s_{j+1}}[u,x]$ contains an edge of $\beta_i$ (by Corollary \ref{cor:path-cross-bisector}). Since the lengths of $P_{u,s_i}[u,x]$ and $P_{u,s_{j+1}}[u,x]$ are the same, $P_{u,s_{j+1}}[u,x] \circ P_{u,s_i}[x,s_i]$ is a shortest $u$-to-$s_{i}$ path. However, since $P_{u,s_{j+1}}[u,x]$ contains an edge of $\beta_i$, we get a contradiction to Corollary \ref{cor:path-disjoint-bisector}.
\end{proof}

\begin{lemma}
\label{lem:bounded-crossing}
Two bisectors can cross at most $\frac{k}{2} + O(1)$ times.
\end{lemma}
\begin{proof}
We will prove that if two bisectors $\beta_i , \beta_j$ cross $r$ times then there exists a vertex $v\in V$ such that $ d(v,s_j)-d(v,s_{i+1}) \ge 2r - \frac{k}{2} - O(1)$. Since the distance between any pair of vertices along the infinite face is at most $\frac{k}{2}$, then by the triangle inequality we have also $d(v,s_j)-d(v,s_{i+1}) \leq \frac{k}{2} $. Hence, we get $2r-\frac{k}{2}-O(1) \leq \frac{k}{2}$ and the lemma follows.

\begin{figure}[htb]
\centering
\subfloat{\includegraphics[scale=0.8]{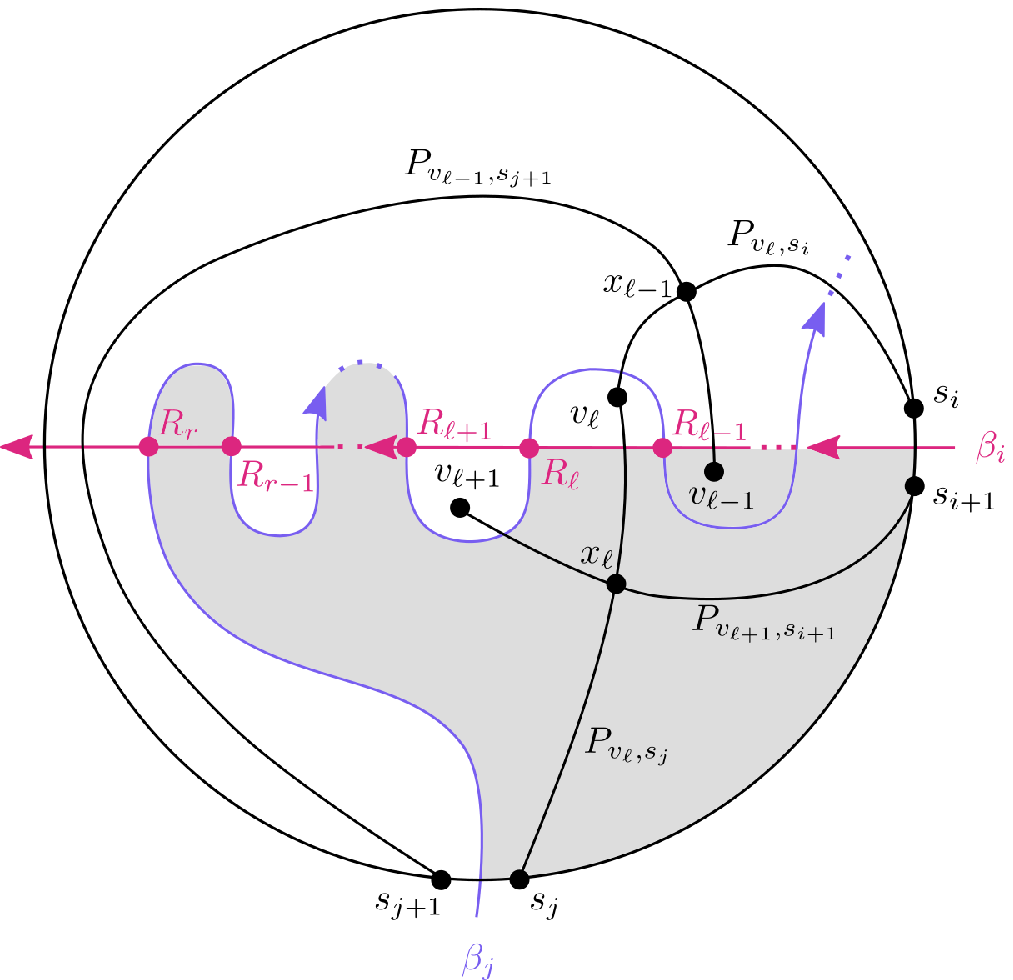}}
\caption{Two bisectors that cross $r$ times and the vertices and shortest paths they induce. Observe that $v_\ell$ is in ${\sf right}(\beta_i)$ when $r-\ell$ is even, and otherwise in ${\sf left}(\beta_i)$. ${\sf right}(C_\ell^*)$ is shaded.\label{r-crossings}}
\end{figure}
Again, let us assume without loss of generality that $s_j$ is in ${\sf left}(\beta_i)$ (the proof of the other case is symmetric). For every $\ell < r$, consider the cycle $C_\ell^* = \beta_j[f_\infty,R_\ell] \circ rev(\beta_i[f_\infty,R_\ell])$. By the previous lemma, $C_\ell^*$ does not self-cross. For even (resp. odd) $r-\ell$, let $(u_\ell v_\ell)^*$ (resp. $(v_\ell u_\ell)^*$) be the arc of $\beta_j$ that follows $R_\ell$. See Figure \ref{r-crossings}. We assume that $r-\ell$ is even (the odd case is symmetric). 

We claim that 
	$C_\ell^*, v_\ell , v_{\ell + 1}$ forms the configuration of Lemma \ref{lem:crossing-config}. By definition of $\beta_j$ we have that $v_\ell \in V \setminus A_j$.
Thus, it remains to prove that $v_{\ell+1} \in A_i$ and that $v_\ell,v_{\ell+1} \in {\sf left}(C_\ell^*)$. By Lemma \ref{lem:opposite-orientation} and the assumption that $s_j \in {\sf left}(\beta_i)$, $\beta_j$ reaches $R_{\ell+1}$ from ${\sf right}(\beta_i)$. Therefore, $( v_{\ell+1} u_{\ell+1})^* \in {\sf left}(\beta_i) \setminus \beta_i$. Hence, $v_{\ell+1}$ is in ${\sf left}(\beta_i)$ and therefore $v_{\ell+1} \in A_i$. To see that $v_\ell,v_{\ell+1} \in {\sf left}(C_\ell^*)$, note that by Lemma \ref{lem:opposite-orientation}, $(v_{\ell+1} u_{\ell+1} )^*$ is in $\beta_j[f_\infty,R_\ell]$, hence $v_{\ell+1} \in {\sf left}(C_\ell^*)$ by definition (as $\beta_j[f_\infty,R_\ell]$ is part of $C_\ell^*$). 
Clearly, since $\beta_j$ enters $R_\ell$ from ${\sf right}(\beta_i)$ we also have that $(u_\ell v_\ell)^* \in {\sf left}(C_\ell^*)$ and in particular $v_\ell \in {\sf left}(C_\ell^*)$.

We can thus apply Lemma \ref{lem:crossing-config} and conclude that $P_{v_{\ell},s_j}$ and $P_{v_{\ell+1},s_{i+1}}$ must intersect at some vertex $x_{\ell}$. 
In addition, we can conclude that $P_{v_{\ell-1},s_{j+1}}$ and $P_{v_{\ell},s_{i}}$ intersect at some vertex $x_{\ell-1}$, and that $P_{v_{\ell-2},s_j}$ and $P_{v_{\ell-1},s_{i+1}}$ intersect at some vertex $x_{\ell-2}$. Therefore, by the triangle inequality and the assumption that patterns are binary, we have:
\begin{equation*}
\begin{split}
&d(v_\ell,x_\ell) + d(x_\ell,s_j) \leq d(v_\ell,x_{\ell-1}) + d(x_{\ell-1},s_{j+1}) -1 \\
&d(v_\ell,x_{\ell-1}) + d(x_{\ell-1},s_i) \leq d(v_\ell,x_\ell) + d(x_\ell,s_{i+1}) -1 \\
&d(v_{\ell-1},x_{\ell-2}) + d(x_{\ell-2},s_{i+1}) \leq d(v_{\ell-1},x_{\ell-1})+d(x_{\ell-1},s_i) -1 \\
&d(v_{\ell-1},x_{\ell-1}) + d(x_{\ell-1},s_{j+1}) \leq d(v_{\ell-1},x_{\ell-2}) + d(x_{\ell-2},s_j) -1
\end{split}
\end{equation*}
Summing the above inequalities we get:
\begin{equation}
\label{eq:crossing}
d(x_\ell,s_j) - d(x_\ell,s_{i+1}) + 4 \leq d(x_{\ell-2},s_j) - d(x_{\ell-2},s_{i+1})
\end{equation}
Let $m \in \{1,2\}$ be so that $r-m$ is even. Let $v = x_m$. By repeating Equation \ref{eq:crossing} we get:

$
d(x_{r-2},s_j)-d(x_{r-2},s_{i+1}) + 4 \cdot \left( \frac{r}{2}-O(1) \right ) \leq d(v,s_j)-d(v,s_{i+1}).
$

\noindent Since the distance between any two vertices along the infinite face is at most $\frac{k}{2}$, it follows that:
$
-\frac{k}{2} + 4 \cdot \left( \frac{r}{2}-O(1) \right ) \leq d(v,s_j)-d(v,s_{i+1}) \leq \frac{k}{2}.
$
Therefore $r \leq \frac{k}{2} + O(1)$.
\end{proof}

\section{A $\Theta(k^2)$ Proof for Halin Graphs}
\label{section:halin}
In this section we prove Theorem~\ref{thm:Halin}. Namely, we show a tight $\Theta(k^2)$ bound on the number of patterns in a family of graphs that includes Halin graphs. 
The {\em Halin graph} family (see~\cite{syslo1983halin} for history and properties) is a restricted family of planar graphs. A Halin graph is obtained from an embedded tree $\mathcal{T}$ with no degree-$2$ vertices by attaching a cycle $C$ to its leaves in their order of appearance according to the embedding. The cycle is then the boundary of the infinite face, and we denote its size by $k$. 
 We will consider a more general family than Halin graphs. Namely, we allow the tree $\mathcal{T}$ to have degree $2$ vertices, and we allow the cycle $C$ to contain vertices that are not in $\mathcal{T}$. We will refer to such graphs as \emph{S-Halin} graphs. See Figure~\ref{halin-8patterns}.

\begin{figure}[htb]
\centering
\includegraphics[scale=1]{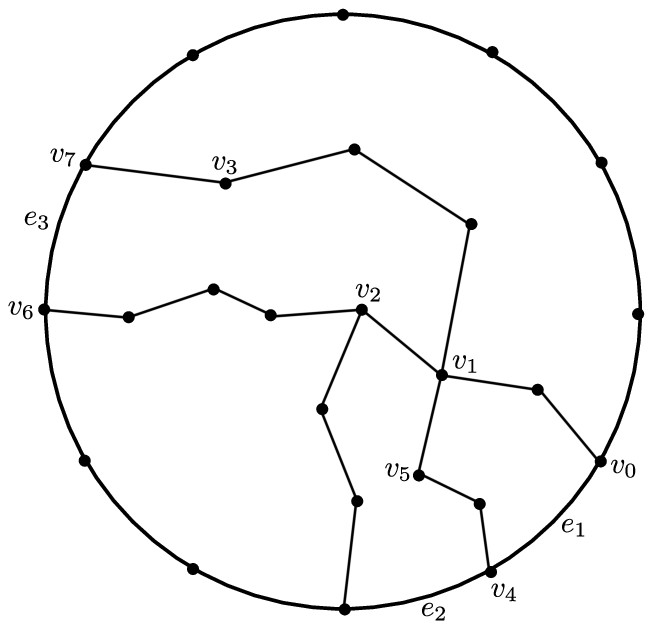}
\caption{An S-Halin graph whose matrix of patterns has VC-dimension $3$.\label{halin-8patterns}}
\end{figure}

 In this section, we show that the number of distinct patterns in S-halin graphs is only $O(k^2)$, and that this bound is tight. In contrast, we show that the VC-dimension argument is limited to proving $O(k^3)$ even on such graphs.

\begin{lemma}
\label{lem:halin-upper-bound}
Let $G = \mathcal{T} \cup C$ be an S-Halin graph, obtained by identifying the leaves of $\mathcal{T}$ with a subset of vertices of a cycle $C$. If the size of $C$ is $k$ then the number of distinct patterns in $G$ is $O(k^2)$.
\end{lemma}
\begin{proof}
Notice that the total number of vertices of $G$ with degree larger than $2$ is $O(k)$ (and hence, the number of faces in $G$ is also $O(k)$). This is because $C$ is of size $k$ and $\mathcal{T}$ contains only $k$ leaves (and hence $O(k)$ vertices with degree larger than $2$). 
Now consider the dual graph $G^* = (V^*,E^*)$. Since every bisector is a simple cycle in $G^*$, and since $|V^*| = O(k)$, then the number of arcs in the graphs $G_\mathcal{B}$ and $G_\mathcal{P}$ is $O(k^2)$. Therefore, by Corollary \ref{cor:pattern-graph} the number of distinct patterns in $G$ is $O(k^2)$.
\end{proof}

We next show that it is not possible to prove Lemma \ref{lem:halin-upper-bound} using the VC-dimension argument. Namely, consider the S-Halin graph of Figure \ref{halin-8patterns} and let $P$ be the matrix whose rows are the patterns of the graph (recall that each pattern is in $\{-1,1\}^{k-1}$). 
 
\begin{proposition}\label{prop:halinVC3}
The VC-dimension of $P$ is $3$.
\end{proposition}
\begin{proof}
Consider the following submatrix of $P$ whose rows correspond to the vertices $v_0, \dots, v_7$ and columns correspond to the edges $e_1,e_2,e_3$:
$$
\begin{blockarray}{cccc}
& e_1 & e_2 & e_3 \\
\begin{block}{c(ccc)}
 v_0 & \ \ 1 & \ \ 1 & \ \ 1\  \\
 v_1 & \ \ 1 & \ \ 1 & -1 \\
 v_2 & \ \ 1 & -1 & \ \ 1 \\
 v_3 & \ \ 1 & -1 & -1 \\
 v_4 & -1 & \ \ 1 & \ \ 1 \\
 v_5 & -1 & \ \ 1 & -1 \\
 v_6 & -1 & -1 & \ \ 1 \\
 v_7 & -1 & -1 & -1 \\
  \end{block}
\end{blockarray}
$$
Since it contains all possible rows, the VC-dimension of $P$ is at least $3$. 
\end{proof}
It is important to remark that we can generalize the example of Figure \ref{halin-8patterns} to any large enough $k$, by adding vertices along the infinite face in the part between $v_7$ and $v_0$ (clockwise).

So far we have seen that S-Halin graphs have at most $O(k^2)$ distinct patterns (Lemma~\ref{lem:halin-upper-bound}), and that the VC-dimension argument is limited to showing $O(k^3)$ distinct patterns (Proposition~\ref{prop:halinVC3}). To conclude this section, we prove that the $O(k^2)$ bound is tight:

\begin{lemma}
\label{lem:halin-lower-bound}
There exists an S-Halin graphs with $\Omega(k^2)$ distinct patterns.
\end{lemma}
\begin{proof}
We assume that $k$ is even and denote $k'=k/2$. We construct the tree $\mathcal{T}$ by taking the union of $k'+1$ simple paths 
 $P_0, P_1, \ldots , P_{k'}$ where every $P_i$ is of length $i$ and all $P_i$'s originate from a common vertex $v_{0,0}$. Namely, $P_i = (v_{i,0} - v_{i,1} - \cdots - v_{i,i})$, and $v_{i,0} = v_{j,0}$ for every $i \neq j$. We choose the embedding of $\mathcal{T}$ so that in a clockwise tour around $v_{0,0}$, the order of appearance of the paths is $(P_1 , P_2 , \cdots , P_{k'})$.

We define an additional $v_{k',k'}$-to-$v_{1,1}$ path of length $k'+1$ denoted $Q = (v_{k',k'}-q_1-q_2 -\cdots - q_{k'}- v_{1,1})$. Let $C$ be the cycle $Q \circ (v_{1,1}-v_{2,2}- \cdots - v_{k',k'})$. Let $G = \mathcal{T} \cup C$. See Figure \ref{halin-quadratic}. Note that $|C|$ is $2k' = k$.
\begin{figure}[htb]
\centering
\includegraphics[scale=1]{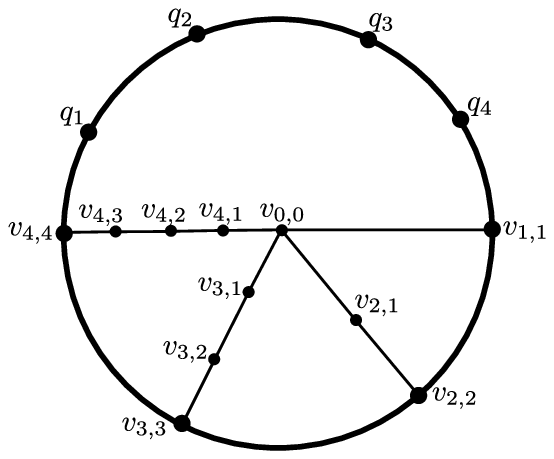}
\caption{The S-Halin of Lemma \ref{lem:halin-lower-bound} for $k=8$.\label{halin-quadratic}}
\end{figure}

Consider the patterns of $G$, when we choose the first vertex to be $v_{1,1}$, the second $v_{2,2}$, etc. Let $i \in [k'],j \in [i-1]$. We claim that the pattern of $v_{i,j}$ is:
\begin{equation}
\label{eq:halin-pattern}
p_{v_{i,j}} = 1^{i-j-1} \circ (-1)^j \circ 1^{k'+j+1-i} \circ (-1)^{k'-j-1}
\end{equation}

To see why, consider the $v_{i,j}$-to-$v_{t,t}$ distances for $t \leq i$. Notice that for every $1 \leq t \leq i-j$ we have $d(v_{i,j},v_{t,t}) = j+t$ and for every $i-j \leq t \leq i$ we have $d(v_{i,j},v_{t,t}) = 2i-j-t$. In particular, they are equal at $t=i-j$. Therefore, the pattern of all the edges between $v_{1,1}$ and $v_{i,i}$ is $1^{i-j-1} \circ (-1)^j$.

Now consider the $v_{i,j}$-to-$v_{t,t}$ distances for $i \leq t \leq k'$. Notice that a shortest $v_{i,j}$-to-$v_{t,t}$ path will never use $Q$ and instead will go through $(v_{i,i}-v_{i+1,i+1}-\cdots-v_{t,t})$. Namely, the distance is $d(v_{i,j},v_{t,t}) = t-j$. Therefore, the pattern of all the edges between $v_{i,i}$ and $v_{k',k'}$ is $1^{k'-i}$.

Finally, consider the $v_{i,j}$-to-$q_t$ distances for $1 \leq t \leq k'$. For every $1 \leq t \leq j+1$ we have $d(v_{i,j},q_t) = k'+t-j$ and for every $j+1 \leq t \leq k'$ we have $d(v_{i,j},q_t) = k'+j+2-t$. In particular they are equal at $t=j+1$. Therefore the pattern of all the edges between $v_{k',k'}$ and $q_k'$ is $1^{j+1} \circ (-1)^{k'-j-1}$. Overall, we get that 
$p_{v_{i,j}}$ is as in Equation (\ref{eq:halin-pattern}).

Notice that $p_{v_{i,j}}$ is unique for every different $i \in [k'],j \in [i-1]$. Since the number of such vertices is $\Omega(k'^2) = \Omega(k^2)$, there are $\Omega(k^2)$ distinct patterns in $G$. 
\end{proof}

\section{Conclusions}
In this work we developed a technique for analyzing the structure and number of distinct patterns in undirected planar graphs. This technique leads to an improved $\tilde{O}(x+k+|T|)$ space compression of the Okamura-Seymour metric, and to an optimal $\tilde{O}(k+|T|)$ compression in the special case where the vertices of $T$ induce a connected component in $G$. Moreover, the technique leads to an alternative proof of the $x=O(k^3)$ upper bound on the number of different patterns.  
 
We have shown that for the family of Halin graphs, the original proof technique using VC-dimension is not tight, and that our approach easily proves the tight bound in this case. Going back to planar graphs, we were unable to come up with constructions of families of planar graphs that have $x=\omega(k^2)$ patterns. We therefore make the following conjecture.

\begin{conjecture*}\label{conjecture}
The number of distinct patterns over all vertices of a planar graph is $O(k^2)$.
\end{conjecture*}

We hope that tools we have developed in this work will be useful in proving this conjecture.

\bibliographystyle{plainurl}

\end{document}